\date{}

\documentclass[11pt]{article}

\title{Beating $1-\frac 1 e$ for Ordered Prophets}

\newcommand{\E}{\mathbb{E}}

\usepackage{makeidx}
\usepackage{euscript}
\usepackage{amsmath,amssymb,amsfonts}
\usepackage{dsfont}
\usepackage{latexsym}
\usepackage{subfigure}
\usepackage{graphicx}
\usepackage{times}\usepackage[scaled=0.92]{helvet} 
\usepackage{euler}
\usepackage{fancybox}
\usepackage{fullpage}
\usepackage{hyperref}
\usepackage{paralist}

\usepackage{color}
\usepackage{wrapfig}
\usepackage{tikz}

\usepackage[small,compact]{titlesec}

\usepackage{algorithm}

\usepackage{algorithmic}

\newcount\shortyear\newcount\shorthour\newcount\shortminute
\shorthour=\time\divide\shorthour by 60\shortyear=\shorthour
\multiply\shortyear by 60\shortminute=\time\advance\shortminute by
-\shortyear \shortyear=\year\advance\shortyear by -1900

\def\zeit{\number\shorthour:\ifnum\shortminute<10 0\number\shortminute
\else\number\shortminute\fi}

\newenvironment{proof}{\noindent{\bf Proof : \ }}{\hfill$\Box$\par\medskip}

\newtheorem{theorem}{Theorem}

\newtheorem{lemma}[theorem]{Lemma}

\newtheorem{definition}[theorem]{Definition}

\newenvironment{proofof}[1]{\begin{trivlist} \item {\bf Proof
#1:~~}}
  {\qed\end{trivlist}}
\renewenvironment{proofof}[1]{\par\medskip\noindent{\bf Proof of #1: \ }}{\hfill$\Box$\par\medskip}

\newcommand*\samethanks[1][\value{footnote}]{\footnotemark[#1]}

\author{
	Melika Abolhasani \samethanks[1]
	\and
	Soheil Ehsani \thanks{Department of Computer Science, University of Maryland, College Park, MD 20742 USA. Email: \texttt{\{melika,ehsani,hossein,hajiagha\}@cs.umd.edu}. Supported in part by NSF CAREER award CCF-1053605, NSF BIGDATA grant IIS-1546108, NSF AF:Medium grant CCF-1161365, DARPA GRAPHS/AFOSR grant FA9550-12-1-0423, and another DARPA SIMPLEX grant.} \thanks{Portions of this research were completed while the authors were visitors at the Simons Institute for the Theory of Computing.}
	\and
	Hossein Esfandiari \samethanks[1]
	\and
	MohammadTaghi HajiAghayi \samethanks[1] \samethanks[2]
	\and
	Robert Kleinberg \thanks{Department of Computer Science, Cornell University, Ithaca, NY 14850 USA. Email: {\tt rdk@cs.cornell.edu}. Supported in part by NSF Award CCF-1512964. This research were initiated while the author was a researcher at Microsoft Research New England, and portions of it were completed while the author was a visitor at the Simons Institute for the Theory of Computing.}
	\and
	Brendan Lucier
}

\begin{document}

\maketitle

\begin{abstract}
	Hill and Kertz studied the \emph{prophet inequality} on iid distributions [\textit{The Annals of Probability 1982}]. They proved a theoretical bound of $1 - \frac 1 e$ on the approximation factor of their algorithm. They conjectured that the best approximation factor for arbitrarily large $n$ is $\frac{1}{1+1/e}\simeq 0.731$. This conjecture remained open prior to this paper for over 30 years.
	
	In this paper we present a threshold-based algorithm for the prophet inequality with $n$ iid distributions. Using a nontrivial and novel approach we show that our algorithm is a $0.738$-approximation algorithm. By beating the bound of $\frac{1}{1+1/e}$, this refutes the conjecture of Hill and Kertz.
	
	Moreover, we generalize our results to non-iid distributions and discuss its applications in mechanism design.
\end{abstract}


\section{Introduction}
Online auctions play a major role in modern markets. In online markets, information about customers and goods is revealed over time.  Irrevocable decisions are made at certain discrete times, such as when a customer arrives to the market. One of the fundamental and basic tools to model this scenario is the \emph{prophet inequality} and its variants.

In a prophet inequality instance we are given a sequence of distributions. Iteratively, we draw a value from one of the distributions, based on a predefined order. In each step we face two choices, either we accept the value and stop, or we reject the value and move to the next distribution. The goal in this problem is to maximize the expected value of the item selected. We say an algorithm for a prophet inequality instance is an $\alpha$-approximation, for $\alpha \leq 1$, if the expectation of the value picked by the algorithm is at least $\alpha$ times that of an optimum solution which knows all of the values in advance.

Prophet inequalities were first studied in the 1970's by Krengel and Sucheston~\cite{kennedy1987prophet,KS77,KS78}. Hajiaghayi, Kleinberg and Sandholm~\cite{HKS07} studied the relation between online auctions and prophet inequalities. In particular they showed that algorithms used in the derivation of prophet inequalities can be reinterpreted as truthful mechanisms for online auctions. Later Chawla, Hartline, Malec, and Sivan~\cite{CHMS10} used prophet inequalities to design sequential posted price mechanisms whose revenue approximates that of the Bayesian optimal mechanism.

In the classical definition of the prophet inequality, the values can be drawn from their distributions in an arbitrary (a.k.a.\ \emph{adversarial} or \emph{worst}) order.  Assuming an adversarial order, the problem has a $0.5$ approximation algorithm which is tight. Recently, Yan~\cite{DBLP:conf/soda/Yan11} considered a relaxed version of this problem in which the algorithm designer is allowed to pick the order of distributions (a.k.a.\ \emph{best order}), and provided a $1-\frac 1 e$ approximation algorithm for this problem. Later, Esfandiari, Hajiaghayi, Liaghat and Monemizadeh \cite{esfandiari2015prophet} showed that there exists a $1-\frac 1 e$ approximation algorithm even when the distributions arrive in a \emph{random order}. Both results provided by Yan and Esfandiari et al.\ are not tight. 

In this work we consider prophet inequalities in both best order and random order settings and take steps towards providing tight approximation algorithms for these problems. Particularly, we consider this problem assuming a large market assumption (i.e.\ we have several copies of each distribution). Indeed, the large market assumption is well-motivated in this context~\cite{buchbinder-jain-naor,devanur-hayes,esfandiari2015online,MSVV,mirrokni2012simultaneous}.

\subsection{Our Contribution} 
First we consider the prophet inequality on a set of \emph{identical and independent distributions} (\emph{iid}). The prophet inequality on iid distributions has been previously studied by Hill and Kertz~\cite{hill1982comparisons} in the 1980's.
Hill and Kertz provided an algorithm based on complicated recursive functions. They proved a theoretical bound of $1 - \frac 1 e$ on the approximation factor of their algorithm, and used a computer program to show that their algorithm is a $0.745$-approximation when the number of distributions is $n=10000$. They conjectured that the best approximation factor for arbitrarily large $n$ is $\frac{1}{1+1/e}\simeq 0.731$. This conjecture remained open for more than three decades. 

In this paper we present a simple threshold-based algorithm for the prophet inequality with $n$ iid distributions. Using a nontrivial and novel approach we show that our algorithm is a $0.738$-approximation algorithm for large enough $n$, beating the bound of $\frac{1}{1+1/e}$ conjectured by Hill and Kertz. This is the first algorithm which is theoretically proved to have an approximation factor better than $1-\frac 1 e$ for this problem. Indeed, beating the $1-\frac 1 e$ barrier is a substantial work in this area~\cite{feige2006approximation,feldman2009online}. The following theorem states our claim formally.

\begin{theorem} \label{intro:theorem:iidResult}
	There exists a constant number $n_0$ such that for every $n\geq n_0$, there exists a  $0.738$-approximation algorithm for any prophet inequality instance with $n$ iid distributions.
\end{theorem}

Next, we extend our results to support different distributions. However, we assume that we have several copies of each distribution. This can be reinterpreted as a large market assumption. 
We say a multiset of independent distributions $\{F_1,\ldots,F_n\}$ is \emph{$m$-frequent} if for each distribution $F_i$ in this multiset there are at least $m$ copies of this distribution in the multiset.
We show that by allowing the algorithm to pick the order of the distributions, there exists a $0.738$-approximation algorithm for any prophet inequality
instance on a set of $m$-frequent distributions, for large enough $m$.
The following theorem states this fact formally.

\begin{theorem}\label{intro:theorem:bestOrder}
	There exists a constant number $m_0$ such that there exits a $0.738$-approximation best order algorithm for any prophet inequality instance on a set of $m_0$-frequent distributions.
\end{theorem}

Our next theorem shows that even in the random order setting one can achieve a $0.738$-approximation algorithm on $m$-frequent distributions, for large enough $m$.

\begin{theorem}\label{intro:theorem:randOrder}
	There exists a constant number $c_0$ such that there exits a $0.738$-approximation random order algorithm for any prophet inequality instance on a set of $(c_0\log(n))$-frequent distributions.
\end{theorem}

To conclude the presentation of our results we show that it is not possible to extend our results to the worst order setting. The following theorem states this fact formally.

\begin{theorem} \label{theorem:hardness}
	For any arbitrary $m$, there is a prophet inequality instance on a set of $m$-frequent distributions such that the instance does not admit any $0.5+\epsilon$-approximation worst order  algorithm.
\end{theorem}

\subsection{Applications in Mechanism Design}
The prophet inequality has numerous applications in 
mechanism design and optimal search theory, so our improved
prophet inequality for $m$-frequent distributions 
has applications in those areas as well. By way of 
illustration, we present here an application to 
optimal search theory. In Weitzman's~\cite{weitzman79}
``box problem'', there are $n$ boxes containing
indepedent random prizes, $v_1,\ldots,v_n$, 
whose distributions are not necessarily identical. 
The cost of opening box $i$ is $c_i \geq 0$. A decision
maker is allowed to open any number of boxes, 
after which she
is allowed to claim the largest prize among the open
boxes. The costs of the boxes, and the distributions 
of the prizes inside, are initially known to the
decision maker, but the value $v_i$ itself is 
only revealed when box $i$ is opened. A search policy
is a (potentially adaptive) rule for deciding
which box to open next---or whether to stop---given
the set of boxes that have already been opened and
the values of the prizes inside. Weitzman~\cite{weitzman79}
derived the structure of the optimal search policy,
which turns out to be wonderfully simple: one computes
an ``option value'' $\sigma_i$ for each box $i$, 
satisfying the equation $E[\max\{0,v_i-\sigma_i\}]=c_i$.
Boxes are opened in order of decreasing $\sigma_i$ 
until there is some open box $i$ such that $v_i > \sigma_j$
for every remaining closed box $j$, then the policy stops.
Kleinberg, Waggoner, and Weyl~\cite{KWW-16} presented
an alternative proof of this result which works by 
relating any instance of the box problem to a modified
instance in which opening boxes is cost-free, but the 
prize in box $i$ is $\min\{v_i,\sigma_i\}$ rather than
$v_i$. The proof shows that when we run any policy on
the modified instance, its net value 
(prize minus combined cost) weakly improves, and that
the net value is preserved if the policy is {\em non-exposed},
meaning that whenever it opens a box with $v_i > \sigma_i$,
it always claims the prize inside. 

An interesting variant of the box problem arises if one
constrains the decision maker, upon stopping, to choose 
the prize in the most recently opened box, rather than
the maximum prize observed thus far. In other words, upon
opening box $i$ the decision maker must irrevocably decide
whether to end the search and claim prize $v_i$, or 
continue the search and relinquish $v_i$. Let us call
this variant the {\em impatient box problem}. It could
be interpreted as modeling, for example, the decision
problem that an employer faces when scheduling a sequence
of costly job interviews in a labor market where hiring 
decisions must be made immediately after the interview.
The factor $1 - \frac 1 e$ prophet inequality of Yan
and Esfandiari et al.\ implies that if the decision maker
is allowed to choose the order in which to inspect boxes
(or even if a random order is used), the net value of the 
optimal impatient box problem policy is at least 
$1 - \frac 1 e$ times the net value of the optimal
policy for the corresponding instance of the original
(non-impatient) box problem; for the proof of this implication,
see Corollary 3 and Remark 1 in~\cite{KWW-16}. A consequence
of Theorem~\ref{intro:theorem:bestOrder} above is that 
this ratio improves to $0.738$ if the instance of the 
impatient box problem contains sufficiently many copies
of each type of box.

Our results also have applications to a recent line of work that employs prophet inequalities to design posted-price mechanisms.
In the standard posted-price setup, a seller has a collection of resources to distribute among $n$ buyers.  The buyers' values are drawn independently from distributions that are known in advance to the seller.  The seller can use this distributional knowledge to set a (possibly adaptive) price on the goods for sale.  Buyers then arrive sequentially and make utility-maximizing purchases.  Hajiaghayi et al.~\cite{HKS07} noted the close connection between this problem and the prophet inequality, with the price corresponding to a choice of threshold.  This has immediate implications for designing prices for welfare maximization, and one can additionally obtain bounds for revenue by applying the prophet inequality to virtual welfare~\cite{ChawlaHK07,CHMS10}.  There has subsequently been a significant line of work extending this connection to derive posted-price mechanisms for broader classes of allocation problems, such as matroid constraints~\cite{kleinberg2012matroid}, multi-item auctions~\cite{A11,CHMS10} and combinatorial auctions~\cite{FeldmanGL15}.
The result of Yan and Esfandiari et al.~\cite{esfandiari2015prophet} implies that for the original case of a single item for sale, if the seller is allowed to choose the order in which the buyers arrive (or if they can be assumed to arrive in random order), then a posted-price mechanism can obtain expected welfare that is at least
$1 - \frac 1 e$ times the expected welfare of the optimal assignment.  Theorem~\ref{intro:theorem:bestOrder} implies that this ratio improves to $0.738$ if the pool of buyers contains sufficiently many individuals whose values are drawn from the same distribution.

\subsection{Other Related Work}
The first generalization of the prophet inequality is the \emph{multiple-choice prophet inequality} ~\cite{kennedy1987prophet,kennedy1985optimal,kertz1986comparison}. In the multiple-choice prophet inequality we are allowed to pick $k$ values, and the goal is to maximize the total sum of picked values. Alaei~\cite{A11} gives an almost tight $(1-\frac{1}{\sqrt{k+3}})$-approximation algorithm for  the $k$-choice prophet inequality (the lower bound is proved in Hajiaghayi, Kleinberg, and Sandholm~\cite{HKS07}).

Prophet inequalities have been studied under complicated combinatorial structures such as matroid, polymatroid, and matching. Kleinberg and Weinberg~\cite{kleinberg2012matroid} consider matroid prophet inequalities, in which the set of selected values should be an independent set of a predefined matroid. They give a tight $0.5$-approximation worst order algorithm for this problem. Later, D\"{u}tting and Kleinberg extended this result to polymatroids~\cite{dutting2015polymatroid}.

Alaei, Hajiaghayi, and Liaghat study matching prophet inequalities~\cite{alaei2011adcell,alaei2013online,DBLP:conf/sigecom/AlaeiHL12}. They extend the multiple-choice prophet inequality and give an almost tight $(1-\frac{1}{\sqrt{k+3}})$-approximation worst order algorithm for any matching prophet inequality instance, where $k$ is the minimum capacity of a vertex. 

Rubinstein considers the prophet inequalities restricted to an arbitrary downward-closed set system~\cite{rubinstein2016beyond}. He provides an $O(\log n \log r)$-approximation algorithm for this problem, where $n$ is the number of distributions and $r$ is the size of the largest feasible set. Babaioff, Immorlica and Kleinberg show a lower bound of $\Omega(\frac{\log n}{\log\log n)} )$ for this problem~\cite{babaioff2007matroids}.
Prophet inequalities has also been studied on many classic problems in graphs~\cite{gobel2014online, garg2008stochastic, dehghani2015online, dehghani2016stochastic}.

\section{IID Distributions} \label{section:iid}
In this section we give a 0.738-approximation algorithm for prophet inequality with iid items. Let us begin with some definitions. Assume that $X_1,\ldots,X_n$ are iid random variables with common distribution function $F$. For simplicity, assume that $F$ is continuous and strictly increasing on a subinterval of $\mathcal{R}^{\geq 0}$. An algorithm based on a sequence of thresholds $\theta_1,\ldots, \theta_n$ is the one that selects the first item $k$ such that $X_k\geq \theta_k$.

Let $\tau$ denote the stopping time of this algorithm, where $\tau$ is $n+1$ when the algorithm selects no item. For simplicity suppose $X_{n+1}$ is a zero random variable. The approximation factor of an algorithm based on $\theta_1\ldots,\theta_n$ is defined as $E[X_\tau]/E[\max X_i]$. This factor captures the ratio between what a player achieves in expectation by acting based on these thresholds and what a prophet achieves in expectation by knowing all $X_i$'s in advance and taking the maximum of them.

In Algorithm \ref{alg:iid} we presents a simple oblivious algorithm for every $n$ and distribution function $F$. Theorem \ref{theorem:iidResult} proves that this algorithm is at least 0.738-approximation for large enough number of items.

\begin{algorithm}
	\normalsize
	\caption{}
	\textbf{Input: $n$ iid items with distribution function $F$.}
	
	\begin{algorithmic}[1] \label{alg:iid}
		\STATE Set $a$ to 1.306 (root of $\cos(a)-\sin(a)/a-1$).
		\STATE Set $\theta_i=F^{-1}(\cos(ai/n)/\cos(a(i-1)/n))$.
		\STATE Pick the first item $i$ for which $X_i\geq \theta_i$.
	\end{algorithmic}
\end{algorithm}

\begin{theorem} \label{theorem:iidResult}
	For every $\epsilon>0$ there exists a number $n_\epsilon$ (a function of $\epsilon$ and independent of $n$) such that for every $n\geq n_\epsilon$ Algorithm \ref{alg:iid} for $n$ items is at least $(1-\epsilon)\alpha$-approximation where $\alpha=1-\cos(a) \approx 0.7388$.
\end{theorem}

In the following we walk you through the steps of the design of Algorithm \ref{alg:iid} and provide a proof for Theorem \ref{theorem:iidResult}. For a given sequence of thresholds let $q_0,q_1,\ldots,q_n$ denote the probability of the algorithm not choosing any of the first items. More specifically, let $q_i=Pr[\tau>i]$ for every $0\leq i \leq n$. Knowing the thresholds $\theta_1,\ldots,\theta_n$ one can find this sequence by starting from $q_0=1$ and computing the rest using $q_{i}=q_{i-1} F(\theta_{i})$. Inversely, one can simply find the thresholds from $q_1,\ldots,q_n$ using $\theta_i=F^{-1}(q_i/q_{i-1})$. Hence, we focus the design of our algorithm on finding the sequence $q_1,\ldots,q_n$. To this end, we aim to find a continuous function $h:[0,1]\rightarrow[0,1]$ with $h(0)=1$ such that by setting $q_i=h(i/n)$ we can achieve our desired set of thresholds.

Note that such a function $h$ has to meet certain requirements. For instance, it has to be strictly decreasing, because at every step the algorithm picks an item with some positive probability, therefore $h(i/n)=q_i=Pr[\tau>i]$ is smaller for larger $i$. In the following we define a class of functions which has two additional properties. We prove that these properties can be useful in designing a useful threshold algorithm.

\begin{definition}
	A continuous and strictly decreasing function $h:[0,1]\rightarrow[0,1]$ with $h(0)=1$ is a threshold function if it has the following two properties:
	\begin{enumerate}[i.]
		\item $h$ is a strictly concave function.
		\item For every $\epsilon>0$ there exists some $\delta_0\leq \epsilon$ such that for every $\delta\leq\delta_0$ and $\enspace\epsilon+\delta\leq s\leq 1$:$$\enspace \frac{h'(s-\delta)}{h(s-\delta)}\leq(1-\epsilon)\frac{h'(s)}{h(s)}\enspace .$$
	\end{enumerate}
\end{definition}

As shown in the following lemma, the first property leads to a decreasing sequence of thresholds. Also, we exploit the second property to show that the approximation factor of $h$ improves by increasing the number of items.

\begin{lemma} \label{lemma:decThresholds}
	If $h$ is a threshold function, then the sequence of thresholds $\theta_1,\ldots,\theta_n$ achieved from $h$ is decreasing.
\end{lemma}
\begin{proof}
	For every $1\leq i\leq n$ we have $\theta_i=F^{-1}(q_i/q_{i-1})$. Since every $q_i=h(i/n)$, we have $\theta_i=F^{-1}(\frac{h(i/n)}{h((i-1)/n)})$. Note that $F$ is a strictly increasing function, therefore having $\theta_i>\theta_{i+1}$ requires $\frac{h(i/n)}{h((i-1)/n)}>\frac{h((i+1)/n)}{h(i/n)}$. For simplicity let $x=i/n$ and $\delta=1/n$. From the first property of threshold functions we have:
	$$h(x) > \frac{h(x+\delta)+h(x-\delta)}{2}\enspace .$$
	By raising both sides to the power of 2, and subtracting $(h(x+\delta)/2-h(x-\delta)/2)^2$ from each side we have:
	\begin{align*}
		h(x)^2-\bigg(\frac{h(x+\delta)-h(x-\delta)}{2}\bigg)^2 &> \bigg(\frac{h(x+\delta)+h(x-\delta)}{2}\bigg)^2 \\ &- \bigg(\frac{h(x+\delta)-h(x-\delta)}{2}\bigg)^2\\
		&=h(x+\delta)h(x-\delta)\enspace .
	\end{align*}	
	Therefore $h(x)^2 > h(x+\delta)h(x-\delta)$, which means $h(x)/h(x-\delta)>h(x+\delta)/h(x)$ and the proof is complete.
\end{proof}

Next, we define a class of functions and prove for every function of this class that its approximation factor approaches $\alpha$ for a large enough $n$. This enables us to narrow down our search for a efficient function $h$.

\begin{definition}
	A threshold function $h$ is $\alpha$-strong if it has the following properties:
	\begin{enumerate}[i.]
		\item $h(1)\leq 1-\alpha$.
		\item $\int_{0}^{1}h(r)dr\geq \alpha$.
		\item $\forall\enspace 0\leq s\leq 1:1-h(s)-\frac{h'(s)}{h(s)}\int_{s}^{1}h(r)dr\geq \alpha(1-\exp(\frac{h'(s)}{h(s)}))\enspace .$
	\end{enumerate}	
\end{definition}

The following theorem formally states our claim for $\alpha$-strong functions.

\begin{theorem} \label{Theorem:alphaStrong}
	If $h$ is an $\alpha$-strong function, then for every $\epsilon>0$ there exists an $n_\epsilon$ such that for every $n\geq n_\epsilon$ the threshold algorithm that acts based on $h$ is at least $(1-\epsilon)\alpha$-approximation on $n$ iid items.
\end{theorem}

\begin{proof}
	Let $OPT$ be a random variable that denotes the optimum solution and $ALG$ be a random variable that denotes the value picked by the algorithm. We can write the expectation of $OPT$ as
	\begin{align}
		\E[OPT]=\int_{0}^{\infty}Pr[\max X_i\geq x]dx \enspace . \label{ineq:integralOpt}
	\end{align}
	Similarly the expectation of $ALG$ is
	\begin{align}
		\E[ALG]=\int_{0}^{\infty}Pr[X_\tau\geq x]dx \enspace .\label{ineq:integralAlg}
	\end{align}
	The main idea behind the proof is to show for $\alpha$-strong functions that the integrand in \eqref{ineq:integralAlg} is an approximation of the integrand in \eqref{ineq:integralOpt} for every non-negative value of $x$. In particular, for every $\epsilon$ there exists some $n_\epsilon$ such that for every $n\geq n_\epsilon$ the second integrand is at least $(1-\epsilon)\alpha$ times the first integrand and this proves the theorem.
	
	Let us begin with finding an upper bound for the integrand in \eqref{ineq:integralOpt}. Let $G(x)=1-F(x)$ for every $x\in \mathcal{R}^{\geq 0}$. The following lemma gives an upper bound for $Pr[\max X_i\geq x]$ based on $G(x)$ and $n$.
	
	\begin{lemma} \label{lemma:OPTupperbound}
		For every $\epsilon>0$ there exists an $n_\epsilon$ such that for every $n\geq n_\epsilon$ the following inequality holds :$$Pr[\max X_i\geq x] \leq \frac{1-\exp(-nG(x))}{1-\epsilon}\enspace .$$
	\end{lemma}
	
	Lemma \ref{lemma:OPTupperbound} gives us an upper bound on $Pr[\max X_i\geq x]$. Now we aim to find a lower bound for $Pr[X_\tau\geq x]$. Through these two bounds we are able to find a lower bound on the approximation factor of the algorithm.
	
	In Lemma \ref{lemma:decThresholds} we showed that the thresholds are decreasing. Hence for an $x\in \mathcal{R}^{\geq 0}$, if $x<\theta_n$ then $Pr[X_\tau\geq x]$ is equal to $Pr[X_\tau\geq \theta_n]$ because the algorithm never selects an item below that value. Moreover, $Pr[X_\tau\geq \theta_n]$ is equal to $Pr[\tau\leq n]$ which is equal to $1-Pr[\tau>n]=1-q_n=1-h(1)$. The first property of $\alpha$-strong functions ensures that this number is at least $\alpha$. Since $Pr[\max X_i\geq x]$ is no more than 1, therefore, for every $x<\theta_n$ we have $Pr[X_\tau\geq x]\geq\alpha Pr[\max X_i\geq x]$.
	
	Now suppose $x\in \mathcal{R}^{\geq 0}$ and $x\geq \theta_n$. For $Pr[X_\tau\geq x]$ we have,
	\begin{align}
		Pr[X_\tau\geq x] &= \sum_{i=1}^{n}Pr[X_\tau\geq x|\tau=i]Pr[\tau=i] \nonumber \\
		&=\sum_{i=1}^{n}q_{i-1}(1-F(\max\{\theta_i,x\})) \enspace .\label{ineq:simpleProb} 
	\end{align}

	Since the thresholds are decreasing, there exists a unique index $j(x)$ for which $\theta_{j(x)}>x\geq\theta_{j(x)+1}$. For the sake of simplicity we assume there is an imaginary item $X_0$ for which $\theta_0=\infty$. In this way $j(x)$ is an integer number from 0 to $n-1$. By expanding \eqref{ineq:simpleProb} we have:
	\begin{align}
		Pr[X_\tau\geq x] &= \sum_{i=1}^{n}q_{i-1}(1-F(\max\{\theta_i,x\})) \nonumber \\
		&= \sum_{i=1}^{n}q_{i-1}G(\max\{\theta_i,x\}) \nonumber \\
		&=\sum_{i=1}^{j(x)}q_{i-1}G(\theta_i) + \sum_{i=j(x)+1}^{n}q_{i-1}G(x) \enspace . \label{ineq:expanded}
	\end{align}	
	
	The first sum in \eqref{ineq:expanded} is indeed the probability of selecting one of the first $j(x)$ items, therefore we can rewrite it as $1-q_{j(x)}$. Hence,
	\begin{align}
		\Pr[X_\tau\geq x] &= 1-q_{j(x)}+\sum_{i=j(x)+1}^{n}q_{i-1}G(x) \nonumber \\
		&= 1-q_{j(x)} + n G(x)\sum_{i=j(x)+1}^{n}q_{i-1}\frac{1}{n} \nonumber \\
		&= 1-q_{j(x)} + n G(x)\sum_{i=j(x)+1}^{n}h((i-1)/n)\frac{1}{n} \nonumber \\
		&\geq 1-q_{j(x)} + n G(x)\int_{j(x)/n}^{1}h(r)dr \enspace . \label{ineq:riemann}
	\end{align}
	The integral in \eqref{ineq:riemann} comes from the fact that $h$ is a decreasing function and for such functions the Riemann sum of an interval is an upper bound of the integral of the function in that interval. For simplicity let $s(x)=j(x)/n$. Inequality \eqref{ineq:riemann} can be written as follows:
	\begin{align}
		Pr[X_\tau\geq x]\geq 1-h(s(x))+nG(x)\int_{s(x)}^{1}h(r)dr \enspace . \label{ineq:integralForm}
	\end{align}	
	
	In order to complete the proof of the theorem, we need to show that the right hand side of Inequality \eqref{ineq:integralForm} is an approximation of $Pr[\max X_i\geq x]$. To this end, we use the following lemma.
	
	\begin{lemma} \label{lemma:algLowerbound}
		For every $\epsilon>0$ there exists an $n_\epsilon$ such that for every integer $n\geq n_\epsilon$ the following inequality holds for every $x\geq \theta_n$:
		$$1-h(s(x))+nG(x)\int_{s(x)}^{1}h(r)dr\geq (1-\epsilon)\alpha(1-\exp(-nG(x)))\enspace .$$
	\end{lemma}
	To wrap up the proof of the theorem we combine the results of the previous lemmas. Suppose $n_1$ and $n_2$ are the lower bounds of Lemma \ref{lemma:OPTupperbound} and Lemma \ref{lemma:algLowerbound} for $n$, respectively, such that their inequalities hold for $\epsilon/2$. For every $n\geq n_\epsilon = \max\{n_1,n_2\}$ we have:
	\begin{align}
	Pr[X_\tau\geq x]
	&\geq 1-h(s(x))+nG(x)\int_{s(x)}^{1}h(r)dr & \text{Inequality}\enspace \eqref{ineq:integralForm} \\
	&\geq (1-\frac{\epsilon}{2})\alpha(1-\exp(-nG(x))) & \text{Lemma}\enspace  \ref{lemma:algLowerbound} \\
	&\geq (1-\frac{\epsilon}{2})^2\alpha \enspace Pr[\max X_i\geq x] & \text{Lemma} \enspace\ref{lemma:OPTupperbound}\\
	&\geq (1-\epsilon)\alpha \enspace Pr[\max X_i\geq x] \enspace.\nonumber
	\end{align}
	This shows that for every non-negative value of $x$ the chance of the algorithm in selecting an item with value at least $x$ is an approximation of the corresponding probability for the optimum solution. More specifically, we showed that for every $n\geq n_\epsilon$ and for every $x\geq 0$ the integrand of \eqref{ineq:integralAlg} is a $(1-\epsilon)\alpha$-approximation of the integrand of \eqref{ineq:integralOpt}, hence the theorem is proved.
\end{proof}

Now we have all the materials needed to prove Theorem \ref{theorem:iidResult}. In order to prove the theorem, we show that the function $h(s)=\cos(as)$ is an $\alpha$-strong function, where $a\approx 1.306$ is a root of $\cos(a)+\sin(a)/a-1$ and $\alpha=1-\cos(a)\approx 0.7388$. To this end, we first need to show that this function is a threshold function:
\begin{enumerate}[i.]
	\item To show the concavity of $h$ it suffices to show that its second derivative is negative for every $0<s\leq 1$. Note that $h'(s)=-a \sin(a s)$ and $h"(s)=-a^2\cos(a s)$.
	\item The ratio of $h'(s)/h(s)$ for every $s$ is equal to $-a\tan(a s)$. For every $\epsilon$ we need to show that there exists some $\delta_0\leq \epsilon$ such that for every $\delta\leq \delta_0$ and $\epsilon+\delta\leq s\leq 1$ the following holds:
	$$-a\tan(a (s-\delta)) \leq -(1-\epsilon)a\tan(a s)$$
	or equivalently, by dividing both sides to $-a$ and changing the direction of the inequality we want to have:
	$$\tan(a s-a \delta)) \geq (1-\epsilon)\tan(a s)\enspace .$$
	Note that $\tan(a s)$ is a convex function because $\tan"(a s)=2\tan(a s)\sec^2(a s)\geq 0$ for $0\leq s\leq 1$. For every $0\leq \delta\leq s$ in such functions we have:
	$$\frac{\tan(a s)-\tan(as-a\delta)}{a\delta}\leq \tan'(as)=\sec^2(as)\leq sec^2(a)\enspace .$$
	Therefore,
	$$\tan(a s)\leq \tan(as-a\delta)+a\delta\sec^2(a)\enspace .$$
	By multiplying both sides by $(1-\epsilon)$ and assuming that $\delta\leq\delta_0=\frac{\epsilon\tan(a\epsilon)}{a(1-\epsilon)\sec^2(a)}$ we have:
	\begin{align}
		(1-\epsilon)\tan(as) &\leq (1-\epsilon)(\tan(as-a\delta)+a\delta\sec^2(a)) \nonumber \\
		&\leq \tan(as-a\delta)-\epsilon\tan(as-a\delta)+(1-\epsilon)a\delta\sec^2(a) \nonumber \\
		&\leq \tan(as-a\delta)-\epsilon\tan(as-a\delta)+\epsilon\tan(a\epsilon) \nonumber \\
		&=\tan(as-a\delta)-\epsilon(\tan(a(s-\delta))-\tan(a\epsilon)) \label{ineq:2ndProperty}
	\end{align}
	Note that $\tan(x)$ is an increasing function, therefore for every $s\geq \epsilon+\delta$ Inequality \eqref{ineq:2ndProperty} is less than or equal to $\tan(as-a\delta)$, thus the second property holds as well.
\end{enumerate}
We showed that $h(s)=\cos(a s)$ is a threshold function. Now we prove that this threshold function is also an $\alpha$-strong function. Due to definition $\alpha=1-cos(a)=1-h(1)$, thus the first property holds. Moreover, $\int_{0}^{1}h(r)dr=sin(a)/a$. Again, due to definition $a$ is a root of $cos(a)+sin(a)/a-1$, and thus $sin(a)/a=1-cos(a)=\alpha$. Now we only need to show that the third property of $\alpha$-strong functions holds. To do so, we need to show that:
\begin{align}
	1-\cos(as)+\tan(as)(\sin(a)-\sin(as))\geq\alpha(1-\exp(-a\tan(as)))\enspace . \label{ineq:finalIneq}
\end{align} 
By subtracting $\alpha(1-\exp(-a\tan(as)))$ from both sides and multiplying them by $\cos(a s)$ we have:

\begin{align*}
\cos(a s) &-\cos(as)^2+\sin(a)sin(as)-\sin(as)^2\\&-\alpha\cos(as)+\alpha\cos(as)\exp(-a\tan(as))\geq 0\enspace .
\end{align*}

Note that $\cos^2(as)+\sin^2(as)=1$, therefore the above inequality is equivalent to:
$$(1-\alpha)\cos(as)+\sin(a)\sin(as)+\alpha\cos(as)\exp(-a\tan(as))\geq 1\enspace .$$
Since $\sin(a)/a=1-\cos(a)=\alpha$ we can replace $\sin(a)$ with $\alpha a$. Also, from the relation between trigonometric functions we have $\cos(x)=1/\sqrt{1+\tan^2(x)}$ and $\sin(x)=\tan(x)/\sqrt{1+\tan^2(x)}$. By considering these equalities and assuming that $w=\tan(as)$ the above inequality becomes simplified as follows:
$$\frac{1-\alpha}{\sqrt{1+w^2}}+\frac{\alpha a w}{\sqrt{1+w^2}}+\frac{\alpha \exp(-a w)}{\sqrt{1+w^2}}\geq 1\enspace .$$
By multiplying both sides by $\sqrt{1+w^2}$ and raising them to the power of two, and subtracting $1+w^2$ from both sides we have:
\begin{align*}
	(1-\alpha+\alpha a w + \alpha \exp(-a w))^2 - 1 - w^2 \geq 0
	\enspace .
\end{align*}

Now we use the following lemma to finish the proof.

\begin{lemma} \label{lemma:aw}
	Suppose $A(w)=(1-\alpha+\alpha a w + \alpha \exp(-a w))^2 - 1 - w^2$ where $a\approx 1.306$ is a root of $\cos(a)+\sin(a)/a-1$ and $\alpha=1-\cos(a)\approx 0.7388$. Then for every $0\leq w\leq \tan(a)$ we have $A(w)\geq 0$.
\end{lemma}

Lemma \ref{lemma:aw} shows that this inequality holds for every $0\leq w\leq \tan(a)$. Consequently, Inequality \eqref{ineq:finalIneq} holds for every $0\leq s\leq 1$. This completes the proof that $h(s)=\cos(as)$ is an $\alpha$-strong function for $\alpha\approx0.7388$, since it has all the three properties.
	

\section{Non IID  Distributions} \label{section:largeMarket}
In this section we study more generalized cases of the prophet inequalities problem. Suppose $X_1,\ldots,X_n$ are random variables from distribution functions $F_1,\ldots,F_n$. Similar to Section \ref{section:iid} we assume, for the sake of simplicity, that all distribution functions are continuous and strictly increasing on a subinterval of $\mathcal{R}^+$. The goal of this section is to show improving results for the best order and a random order of \textit{large market} instances. We use the term large market as a general term to refer to instances with repeated distributions. The following definition formally captures this concept.

\begin{definition}
	A set of $n$ items is called $m$-frequent if for every item $i$ with distribution function $F_i$ there are at least $m-1$ other items in the set with the same distribution function as $F_i$.
\end{definition}

In the remainder of this section we show for the best order and a random order of a large market instance that one can find a sequence of thresholds which in expectation performs as good as our algorithm for iid items. Roughly speaking, we design algorithms that are $\alpha$-approximation for large enough $m$-frequent instances, where $\alpha\approx 0.7388$. The following two theorems formally state our results for the best order and a random order, respectively.
 
\begin{theorem}\label{theorem:bestOrder}
	For every $\epsilon>0$ and set $\mathbb{X}$ of $n$ items, there exists a number $m_\epsilon$ (a function of $\epsilon$ and independent of $n$) such that if $\mathbb{X}$ is $m$-frequent for $m\geq m_\epsilon$ then there exits an algorithm which is $(1-\epsilon)\alpha$-approximation on a permutation of $\mathbb{X}$.
\end{theorem}

\begin{theorem}\label{theorem:randOrder}
	For every $\epsilon>0$ and set $\mathbb{X}$ of $n$ items there exists a number $c_\epsilon$ (a function of $\epsilon$ and independent of $n$) such that if $\mathbb{X}$ is $m$-frequent for $m\geq c_\epsilon \log(n)$ then there exists an algorithm which in expectation is $(1-\epsilon)\alpha$-approximation on a random permutation of $\mathbb{X}$.
\end{theorem}

To prove the theorems we first provide an algorithm for a specific class of large market instances, namely partitioned sequences. Lemma \ref{lemma:partition} states that this algorithm is $\alpha$-approximation when the number of partitions is large. We later show how to apply this algorithm on the best order and a random order of large market instances to achieve a similar approximation factor. Following is a formal definition of partitioned sequences.

\begin{definition}
	A sequence of items with distribution functions $F_1,\ldots,F_n$ is $m$-partitioned if $n=mk$ and the sequence of functions $F_{ik+1},\ldots,F_{ik+k}$ is a permutation of $F_1,\ldots,F_k$ for every $0\leq i<m$.
\end{definition}

The following algorithm exploits Algorithm \ref{alg:iid} for iid items in order to find thresholds for a partitioned large market instance. 

\begin{algorithm}
	\normalsize
	\caption{}
	\textbf{Input: An $m$-partitioned sequence of items with distribution functions $F_1,\ldots,F_n$.}
	
	\begin{algorithmic}[1] \label{alg:largeMarket}
		\STATE Let $k=n/m$.
		\STATE Let $F(x)=\prod_{i=1}^{k}F_i(x)$.
		\STATE Let $\theta_1\ldots,\theta_m$ be the thresholds by Algorithm \ref{alg:iid} for $m$ iid items with distribution function $F$.
		\STATE Pick the first item $i$ if $X_i\geq \theta_{\lceil i/k \rceil}$.
	\end{algorithmic}
\end{algorithm}

\begin{lemma} \label{lemma:partition}
	For every $\epsilon>0$ there exists a number $m_\epsilon$ (a function of $\epsilon$ and independent of the number of items) such that for every $m\geq m_\epsilon$ Algorithm \ref{alg:largeMarket} is $(1-\epsilon)\alpha$-approximation on an $m$-partitioned input.
\end{lemma}

Now we are ready to prove Theorem \ref{theorem:bestOrder} and Theorem \ref{theorem:randOrder}.

\begin{proofof} {Theorem \ref{theorem:bestOrder}}
	Let $s$ be the lower bound on the number of partitions in Lemma \ref{lemma:partition} for $\epsilon/2$, and let $m_\epsilon=2(s-1)/\epsilon$. The outline of the proof is as follows. Let $\mathbb{X}$ be an $m$-frequent set of items for $m\geq m_\epsilon$. We uniformly group the items into $s$ parts with $\lfloor m/s \rfloor$ items of each type in every group. Let $\mathbb{Y}$ denote the set of partitioned items. In order to make all parts similar, we may need to discard some of the items, however, we show this does not hurt the approximation factor significantly. Finally, by applying Algorithm \ref{alg:largeMarket} to $\mathbb{Y}$ we achieve the desired approximation factor.
	
	The following lemma shows that discarding a fraction of items influences the approximation factor proportionally.
	\begin{lemma} \label{lemma:discard}
		Let $\{X_1,\ldots,X_n\}$ be a $k$-frequent set of items. Suppose for some $S\subseteq\{1,\ldots,n\}$ that the set $\{X_{S_1},\ldots,X_{S_r}\}$ is $p$-frequent and contains every $X_i$ for $1\leq i\leq n$. Then we have
		$$\E[\max_{i\in S}X_i] \geq \frac{p}{k}\E[\max_{1\leq i\leq n}X_i] \enspace .$$
	\end{lemma}
	
	Note that in partitioning $\mathbb{X}$ to $s$ groups there might be at most $s-1$ items of each type being discarded in $\mathbb{Y}$, therefore $\mathbb{Y}$ is $(m-s+1)$-frequent. Let $ALG$ be a random variable that denotes the value of the item picked by our algorithm. We have:
	\begin{align*}
		\E[ALG] &\geq (1-\frac{\epsilon}{2})\alpha \E[\max_{Y\in\mathbb{Y}}Y] & \text{Lemma} \enspace \ref{lemma:partition} \\
		&\geq (1-\frac{\epsilon}{2})\alpha \frac{m-s+1}{m}\E[\max_{X\in\mathbb{X}}X] & \text{Lemma} \enspace \ref{lemma:discard} \\
		&\geq (1-\frac{\epsilon}{2})^2\alpha \E[\max_{X\in\mathbb{X}}X] & \\
		&\geq (1-\epsilon)\alpha \E[\max_{X\in\mathbb{X}}X] \enspace .
	\end{align*}
	Therefore, for every $m$-frequent set $\mathbb{X}$ there exists an ordering of its items  on which our algorithm is $(1-\epsilon)\alpha$-approximation.
\end{proofof}

\begin{proofof}{Theorem \ref{theorem:randOrder}}
	Let $\pi$ be a random permutation of the items. Consider $s$ different partitions for the items, i.e. one from $X_{\pi_1}$ to $X_{\pi_{n/s}}$, one from $X_{\pi_{n/s+1}}$ to $X_{\pi_{2 n/s}}$, so on so forth. We show that when the number of similar items is large enough then a random permutation is very likely to uniformly distribute similar items into these parts. Therefore, by discarding a small fraction of the items $X_{\pi_1},\ldots,X_{\pi_n}$ can be assumed as an $s$-partitioned sequence, hence Algorithm \ref{alg:largeMarket} can be applied to it.
	
	Note that $\mathbb{X}$ is $m$-frequent, which means that for every item $i$ there are at least $m-1$ other items with the same distribution functions as $F_i$. We refer to a set of similar items as a type. Therefore, there are at least $m$ items of every type in $\mathbb{X}$. We use the following lemma to show for every type that with a high probability the number of items of that type in every partition is almost $m/s$.
	
	\begin{lemma} [\cite{panconesi1997randomized}] \label{lemma:chernoff}
		Let $x_1,\ldots,x_m$ be a sequence of negatively correlated boolean (i.e. 0 or 1) random variables, and let $X=\sum_{i=1}^{m} x_i$. We have:
		$$Pr[|X-\E[X]|\geq \delta \E[X]] \leq 3\exp(-\frac{\delta^2\E[X]}{3})\enspace .$$
	\end{lemma}
	
	Since $\pi$ is a random permutation, the expected number of these items in a fixed partition is $m/s$. Using Lemma \ref{lemma:chernoff}, with probability at most $3\exp(\frac{-\delta^2 m}{3 s})$ there are less than $(1-\delta)m/s$ of these items in a fixed partition. Using Union Bound on all the $s$ partitions and all types of items (note that there at at most $n/m$ types), with probability at most $3 s \frac{n}{m}\exp(\frac{-\delta^2 m}{3 s})$ there is a type of item  which has less than $(1-\delta)m/s$ items in a partitions. If we choose $\delta=\epsilon/3$ then for every $m\geq \frac{3 s}{\delta^2}(\log(n)+\log(\frac{9}{\epsilon}))$ this probability becomes less than $\epsilon/3$.
	
	Now we are ready to wrap up the proof. If we choose $s=m_{\epsilon/3}$ using Lemma \ref{lemma:partition}, $\delta=\epsilon/3$, and $c_\epsilon=\frac{3s\log(9/\epsilon)}{\delta^2}$ then for every $m\geq c_\epsilon\log(n)$ with probability at least $(1-\epsilon/3)$ there are at least $(1-\epsilon/3)m/s$ items of each type in every partitions. In such cases by discarding at most $\epsilon/3$ fraction of the items of each type we have exactly $(1-\epsilon/3)m/s$ of them in each partition. Lemma \ref{lemma:discard} states that removing this fraction of items changes the approximation factor by at most $(1-\epsilon/3)$. This means that for a random permutation of the items, with probability at least $(1-\epsilon/3)$ we can loose on the approximation factor by no worse than $(1-\epsilon/3)$ and have an $s$-partitioned sequence. Due to Lemma \ref{lemma:partition}, Algorithm \ref{alg:largeMarket} is $(1-\epsilon/3)\alpha$-approximation on this number of partitions. Therefore, the approximation factor of our method is $(1-\epsilon/3)^3\alpha$ which is more than $(1-\epsilon)\alpha$.
\end{proofof}

\section{Conclusions and Open Problems}
In this paper we demonstrate a simple algorithm for the iid prophet inequality problem. We analyze our algorithm through a class of functions called $\alpha$-strong, and show that the set of thresholds based on such functions guarantee for every $x$ that the probability of the algorithm picking at least $x$ is no less that $\alpha$ times the probability of the maximum being at least $x$. This simply yields the approximation factor of $\alpha$ for the expectations. Finally by proposing a 0.738-strong function we complete the proof.

One question that arises here is whether this is the best achievable approximation factor for iid prophet inequality. We can approach this problem from two directions, i.e. the lower bound and the upper bound. We believe that our $\cos(a s)$ function is not the strongest due to the gap that exists in meeting property iii of strong functions. Although such a gap suggests the existence of a stronger threshold function, it is interesting to find one that has a closed form representation.



It is also worth noting that the gap does not seem to be large. In other words, Hill an Kertz~\cite{hill1982comparisons} showed using a computer program that the best approximation one can get is 0.748 for $n=100$ and 0.745 for $n=1000$. They conjecture the bound of $\frac{1}{1+1/e}=0.731$ for infinitely large $n$, which is refuted by our results. However, another interesting question is whether for every $n$ there is a distribution function that bounds the performance of any online algorithm for iid prophet inequality.

The main question that we leave open in this paper is whether one can beat the $1-1/e$ barrier for general distributions, namely the prophet secretary problem introduced in \cite{esfandiari2015prophet}. A direction to solve this problem could be similar to our approach for large market inputs. In other words, does there exist a black-box reduction from general distributions to iid instances which preserves the approximation factor? Finally, we would like to note that our approach for iid instances seems to become more complicated for different distributions, hence finding a simpler solution for iid prophet inequality that beats $1-1/e$ would be interesting as well.





\newcommand{\Proc}{Proceedings of the~}
\newcommand{\STOC}{Annual ACM Symposium on Theory of Computing (STOC)}
\newcommand{\FOCS}{IEEE Symposium on Foundations of Computer Science (FOCS)}
\newcommand{\SODA}{Annual ACM-SIAM Symposium on Discrete Algorithms (SODA)}
\newcommand{\SOCG}{Annual Symposium on Computational Geometry (SoCG)}
\newcommand{\ICALP}{Annual International Colloquium on Automata, Languages and Programming (ICALP)}
\newcommand{\ESA}{Annual European Symposium on Algorithms (ESA)}
\newcommand{\CCC}{Annual IEEE Conference on Computational Complexity (CCC)}
\newcommand{\RANDOM}{International Workshop on Randomization and Approximation Techniques in Computer Science (RANDOM)}
\newcommand{\PODS}{ACM SIGMOD Symposium on Principles of Database Systems (PODS)}
\newcommand{\SSDBM}{ International Conference on Scientific and Statistical Database Management (SSDBM)}
\newcommand{\ALENEX}{Workshop on Algorithm Engineering and Experiments (ALENEX)}
\newcommand{\BEATCS}{Bulletin of the European Association for Theoretical Computer Science (BEATCS)}
\newcommand{\CCCG}{Canadian Conference on Computational Geometry (CCCG)}
\newcommand{\CIAC}{Italian Conference on Algorithms and Complexity (CIAC)}
\newcommand{\COCOON}{Annual International Computing Combinatorics Conference (COCOON)}
\newcommand{\COLT}{Annual Conference on Learning Theory (COLT)}
\newcommand{\COMPGEOM}{Annual ACM Symposium on Computational Geometry}
\newcommand{\DCGEOM}{Discrete \& Computational Geometry}
\newcommand{\DISC}{International Symposium on Distributed Computing (DISC)}
\newcommand{\ECCC}{Electronic Colloquium on Computational Complexity (ECCC)}
\newcommand{\FSTTCS}{Foundations of Software Technology and Theoretical Computer Science (FSTTCS)}
\newcommand{\ICCCN}{IEEE International Conference on Computer Communications and Networks (ICCCN)}
\newcommand{\ICDCS}{International Conference on Distributed Computing Systems (ICDCS)}
\newcommand{\VLDB}{ International Conference on Very Large Data Bases (VLDB)}
\newcommand{\IJCGA}{International Journal of Computational Geometry and Applications}
\newcommand{\INFOCOM}{IEEE INFOCOM}
\newcommand{\IPCO}{International Integer Programming and Combinatorial Optimization Conference (IPCO)}
\newcommand{\ISAAC}{International Symposium on Algorithms and Computation (ISAAC)}
\newcommand{\ISTCS}{Israel Symposium on Theory of Computing and Systems (ISTCS)}
\newcommand{\JACM}{Journal of the ACM}
\newcommand{\LNCS}{Lecture Notes in Computer Science}
\newcommand{\RSA}{Random Structures and Algorithms}
\newcommand{\SPAA}{Annual ACM Symposium on Parallel Algorithms and Architectures (SPAA)}
\newcommand{\STACS}{Annual Symposium on Theoretical Aspects of Computer Science (STACS)}
\newcommand{\SWAT}{Scandinavian Workshop on Algorithm Theory (SWAT)}
\newcommand{\TALG}{ACM Transactions on Algorithms}
\newcommand{\UAI}{Conference on Uncertainty in Artificial Intelligence (UAI)}
\newcommand{\WADS}{Workshop on Algorithms and Data Structures (WADS)}
\newcommand{\SICOMP}{SIAM Journal on Computing}
\newcommand{\JCSS}{Journal of Computer and System Sciences}
\newcommand{\JASIS}{Journal of the American society for information science}
\newcommand{\PMS}{ Philosophical Magazine Series}
\newcommand{\ML}{Machine Learning}
\newcommand{\DCG}{Discrete and Computational Geometry}
\newcommand{\TODS}{ACM Transactions on Database Systems (TODS)}
\newcommand{\PHREV}{Physical Review E}
\newcommand{\NATS}{National Academy of Sciences}
\newcommand{\MPHy}{Reviews of Modern Physics}
\newcommand{\NRG}{Nature Reviews : Genetics}
\newcommand{\BullAMS}{Bulletin (New Series) of the American Mathematical Society}
\newcommand{\AMSM}{The American Mathematical Monthly}
\newcommand{\JAM}{SIAM Journal on Applied Mathematics}
\newcommand{\JDM}{SIAM Journal of  Discrete Math}
\newcommand{\JASM}{Journal of the American Statistical Association}
\newcommand{\AMS}{Annals of Mathematical Statistics}
\newcommand{\JALG}{Journal of Algorithms}
\newcommand{\TIT}{IEEE Transactions on Information Theory}
\newcommand{\CM}{Contemporary Mathematics}
\newcommand{\JC}{Journal of Complexity}
\newcommand{\TSE}{IEEE Transactions on Software Engineering}
\newcommand{\TNDE}{IEEE Transactions on Knowledge and Data Engineering}
\newcommand{\JIC}{Journal Information and Computation}
\newcommand{\ToC}{Theory of Computing}
\newcommand{\Algorithmica}{Algorithmica}
\newcommand{\MST}{Mathematical Systems Theory}
\newcommand{\Com}{Combinatorica}
\newcommand{\NC}{Neural Computation}
\newcommand{\TAP}{The Annals of Probability}

\bibliographystyle{abbrv}
\bibliography{References}

\appendix
\newpage
\section{Appendix: Omitted Proofs} \label{appendix1}

\subsection{Proof of Lemma \ref{lemma:OPTupperbound}}
	\begin{proof}
		For every $x\in\mathcal{R}^{\geq 0}$ we have:
		\begin{align}
		Pr[\max X_i\geq x] &= 1-F(x)^n \nonumber \\
		&= 1-(1-G(x))^n \nonumber \\
		&= 1-\bigg(\frac{1}{1-G(x)}\bigg)^{-n} \nonumber \\
		&= 1-\bigg(1+\frac{G(x)}{1-G(x)}\bigg)^{-n} \nonumber \\
		&\leq 1-\exp\bigg(\frac{-nG(x)}{1-G(x)}\bigg)\enspace .
		\end{align}
		We complete the proof by proving for every $\epsilon> 0$ that there exists an $n_\epsilon$ such that for every $n\geq n_\epsilon$ and $0\leq z\leq 1$, the ratio between $A(n,z)=1-\exp(-nz/(1-z))$ and $B(n, z)=1-\exp(-nz)$ is no more than $1/(1-\epsilon)$.
		
		For every $n$ and $z$ there are two cases:
		\begin{itemize}
			\item If $ln(n)/n\leq z\leq 1$ then we have:
			\begin{align}
			\frac{A(n,z)}{B(n,z)} \leq \frac{1}{B(n,z)} \leq \frac{1}{1-\exp(-ln(n))} = \frac{1}{1-1/n}\enspace .\label{ineq:nz1}
			\end{align}
			\item If $0\leq z\leq ln(n)/n$, we use partial derivatives of the functions to find an upper bound of their ratio. In the following the derivative of a function is with respect to variable $z$.
			\begin{align}
			\frac{A(n,z)}{B(n,z)}
			&= \frac{\int_{0}^{z}A'(n, w)dw}{\int_{0}^{z}B'(n, w)dw} \nonumber\\
			&= \frac{\int_{0}^{z}B'(n, w)\frac{A'(n, w)}{B'(n, w)}dw}{\int_{0}^{z}B'(n, w)dw} \nonumber\\
			&\leq \frac{\int_{0}^{z}B'(n, w)dw}{\int_{0}^{z}B'(n, w)dw}.\max_{0\leq w\leq z}\bigg\{\frac{A'(n, w)}{B'(n, w)}\bigg\} \nonumber\\
			&= \max_{0\leq w\leq z}\bigg\{\frac{A'(n, w)}{B'(n, w)}\bigg\} \nonumber\\
			&= \max_{0\leq w\leq z}\bigg\{\frac{n \exp(-nz/(1-z))/(1-z)^2}{n\exp(-nz)}\bigg\} \nonumber\\
			&= \max_{0\leq w\leq z}\bigg\{\frac{\exp(-nz^2/(1-z))}{(1-z)^2}\bigg\} \nonumber\\
			&\leq \frac{1}{(1-z)^2} \nonumber\\
			&\leq \frac{1}{1-2\ln(n)/n+ln^2(n)/n^2}\enspace . \label{ineq:nz2}
			\end{align}
		\end{itemize}
		Note that the denominators of both \eqref{ineq:nz1} and \eqref{ineq:nz2} become greater than $1-\epsilon$ as $n$ becomes greater than some $n_\epsilon$, thus the proof of the lemma follows.
	\end{proof}
	
	\subsection{Proof of Lemma \ref{lemma:algLowerbound}}
	\begin{proof}
		Recall that $s(x)=j(x)/n$ is a number from 0 to 1. We prove the correctness of the lemma by analyzing it for two different ranges of $s(x)$. For simplicity we may refer to $s(x)$ as $s$ in different parts of the proof. Suppose $s_0=min(0.5, \alpha)\epsilon$. For $0\leq s\leq s_0$ we have:
		\begin{align}
		1-h(s)+nG(x)\int_{s}^{1}h(r)dr
		&\geq nG(x)\int_{s}^{1}h(r)dr \nonumber \\
		&\geq nG(x)\int_{s_0}^{1}h(r)dr \nonumber \\
		&= nG(x)(\int_{0}^{1}h(r)dr-\int_{0}^{s_0}h(r)dr) \nonumber \\
		&\geq nG(x)(\int_{0}^{1}h(r)dr-s_0)	\label{ineq:simpleCase}\enspace.
		\end{align}
		From the second property of $\alpha$-strong functions we have $\int_{0}^{1}h(r)dr\geq \alpha$. Also, for every $z\in \mathcal{R}^{\geq 0}$ it holds that $z\geq 1-\exp(-z)$. By using these two inequalities in Inequality \eqref{ineq:simpleCase} the lemma is proved for this case:
		\begin{align*}
		1-h(s)+nG(x)\int_{s}^{1}h(r)dr
		&\geq nG(x)(\int_{0}^{1}h(r)dr-s_0) &\\
		&\geq (1-\exp(-nG(x)))(\alpha-s_0) \nonumber & \\
		&\geq (1-\exp(-nG(x)))\alpha(1-\frac{s_0}{\alpha})\nonumber\\
		&\geq (1-\exp(-nG(x)))\alpha(1-\epsilon)\nonumber & s_0\leq \alpha\epsilon
		\end{align*}
		Now what remains is the case that $s_0<s\leq 1$. Again, for this case we want the following inequality to hold:
		\begin{align}
		\frac{1-h(s)+nG(x)\int_{s}^{1}h(r)dr}{1-\exp(-nG(x))}\geq (1-\epsilon)\alpha\enspace . \label{ineq:mainIneq}
		\end{align}
		Recall that for every $x\geq \theta_n$, $s(x)=j(x)/n$ where $j(x)$ is the greatest index for which $\theta_j>x\geq \theta_{j+1}$. Since $G$ is a strictly decreasing function, we have $G(\theta_j)<G(x)\leq G(\theta_{j+1})$. Recall that for every $1\leq i\leq n$ we have $q_i=q_{i-1}(1-G(\theta_i))$, or equivalently $G(\theta_i)=1-q_i/q_{i-1}$. Therefore we can bound $G(x)$ as follows:
		\begin{align}
		1-\frac{q_j}{q_{j-1}} < G(x) \leq 1-\frac{q_{j+1}}{q_j}\enspace . \label{ineq:boundsGx}
		\end{align}
		Now, finding a lower bound for $1-q_j/q_{j-1}$ and an upper bound for $1-q_{j+1}/q_j$ in Inequality \eqref{ineq:boundsGx} gives us a lower bound and an upper bound for $G(x)$. For the lower bound we have
		\begin{align*}
			G(\theta_j)&=1-\frac{q_j}{q_{j-1}} \\ &=\frac{q_{j-1}-q_j}{q_{j-1}}\\ &=\frac{-(q_j-q_{j-1})}{q_{j-1}}\\ &=\frac{-(h(s)-h(s-1/n))}{h(s-1/n)} \enspace .
		\end{align*}
		By multiplying this fraction by $n/n$ we get:
		\begin{align}
		G(\theta_j)=\frac{-(\frac{h(s)-h(s-1/n)}{1/n})}{n\enspace h(s-1/n)}\geq \frac{-h'(s-1/n)}{n\enspace h(s-1/n)} \enspace , \label{ineq:derivative1}
		\end{align}
		where the last inequality in \eqref{ineq:derivative1} comes from the concavity of $h$. From the second property of threshold functions there exists some $\delta_0\leq s_0$ such that for every $n\geq 1/\delta_0, s_0+1/n\leq s\leq 1$ the following inequality holds:
		\begin{align}
		\frac{-h'(s-1/n)}{h(s-1/n)}\geq (1-s_0)\frac{-h'(s)}{h(s)}\enspace . \label{ineq:derivative2}
		\end{align}
		By using Inequality \eqref{ineq:derivative2} in Inequality \eqref{ineq:derivative1}, and by using that inequality in Inequality \eqref{ineq:boundsGx}, we get:
		\begin{align}
		G(x)>(1-s_0)\frac{-h'(s)}{n\enspace h(s)}\enspace . \label{ineq:lowerbound}
		\end{align}
		Similarly one can show the following upper bound on $G(x)$:
		\begin{align}
		G(x) &\leq G(\theta_{j+1}) & \nonumber \\
		&= 1-\frac{q_{j+1}}{q_{j}} & \nonumber \\
		&= \frac{-(h(s+1/n)-h(s))}{h(s)} & \nonumber \\
		&= \frac{-\frac{h(s+1/n)-h(s)}{1/n}}{n\enspace h(s)} & \text{multiplying by $\frac{n}{n}$} \nonumber \\
		&\leq \frac{-\frac{h(s+1/n)-h(s)}{1/n}}{n\enspace h(s+1/n)} & \text{since $h$ is decreasing} \nonumber \\
		&\leq \frac{-h'(s+1/n)}{n\enspace h(s+1/n)} \label{ineq:useful} & \text{concavity of $h$}\\
		&\leq \frac{1}{1-s_0}.\frac{-h'(s)}{n\enspace h(s)}\enspace. & \text{threshold functions (ii)} \label{ineq:upperbound}
		\end{align}
		Using these bounds and the following auxiliary lemma we prove the correctness of Inequality \eqref{ineq:mainIneq}.
		\begin{lemma} \label{lemma:auxiliary}
			For every $z<0$ and $t\geq 1$ we have: $(1-\exp(zt))/(1-\exp(z))\leq t$.
		\end{lemma}
		\begin{proof}
			Let $A(z)=1-\exp(zt)$ and $B(z)=1-\exp(z)$. In the following the derivatives are with respect to $z$. For the ratio of these functions we have:
			\begin{align*}
			\frac{A(z)}{B(z)} &= \frac{\int_{0}^{z}A'(w)dw}{\int_{0}^{z}B'(w)dw} \\
			&=\frac{\int_{0}^{z}B'(w)\frac{A'(w)}{B'(w)}dw}{\int_{0}^{z}B'(w)dw}\\
			&\leq \max_{z\leq w\leq 0}\bigg\{\frac{A'(w)}{B'(w)}\bigg\}\frac{\int_{0}^{z}B'(w)dw}{\int_{0}^{z}B'(w)dw}\\
			&\leq \max_{z\leq w\leq 0}\bigg\{\frac{-t\exp(tw)}{-\exp(w)}\bigg\} \\
			&=\max_{z\leq w\leq 0}\bigg\{t\exp(w(t-1))\bigg\} \enspace .
			\end{align*}
			Since $w\leq 0$ and $t\geq 1$ the $\exp(w(t-1))\leq 1$, and the proof follows.
		\end{proof}
		Using the bound of Inequality \eqref{ineq:lowerbound} in the left hand side of Inequality \eqref{ineq:mainIneq} and because $1-h(s)\geq 0$ we achieve
		\begin{align*}
		\frac{1-h(s)+nG(x)\int_{s}^{1}h(r)dr}{1-\exp(-nG(x))}
		&\geq \frac{1-h(s)-(1-s_0)\frac{h'(s)}{h(s)}\int_{s}^{1}h(r)dr}{1-\exp(-nG(x))} \\
		&\geq \frac{(1-s_0)(1-h(s)-\frac{h'(s)}{h(s)}\int_{s}^{1}h(r)dr)}{1-\exp(-nG(x))}
		\end{align*}
		By applying Inequality \eqref{ineq:upperbound} we get
		\begin{align*}
		\frac{1-h(s)+nG(x)\int_{s}^{1}h(r)dr}{1-\exp(-nG(x))}	
		&\geq\frac{(1-s_0)(1-h(s)-\frac{h'(s)}{h(s)}\int_{s}^{1}h(r)dr)}{1-\exp(\frac{h'(s)}{h(s)(1-s_0)}))} \enspace .
		\end{align*}
		By applying Lemma \ref{lemma:auxiliary} to the denominator we have:
		\begin{align*}
		\frac{1-h(s)+nG(x)\int_{s}^{1}h(r)dr}{1-\exp(-nG(x))}
		&\geq(1-s_0)^2.\frac{1-h(s)-\frac{h'(s)}{h(s)}\int_{s}^{1}h(r)dr}{1-\exp(h'(s)/h(s)))} \enspace .
		\end{align*}
		From the third property of $\alpha$-strong functions, the fraction at the right hand side of the above inequality is at least $\alpha$. Moreover, since $s_0\leq 0.5\epsilon$ it holds that $(1-s_0)^2\geq(1-\epsilon)$, and thus Inequality \eqref{ineq:mainIneq} holds and the proof of the lemma is complete.
	\end{proof}
	
\subsection{Proof of  Lemma \ref{lemma:aw}}

\begin{proof}
	Let us first take a look at the first three derivatives of $A(w)$ which are all continuous and bounded in range $[0,\tan(a)]$:
	\begin{align*}
		& A'(w)=2\alpha a (1-\exp(-a w))(1-\alpha+\alpha a w + \alpha \exp(-a w))-2w, \\
		& A''(w)=2\alpha a^2 \exp(-aw)(1-3\alpha+\alpha a w + 2\alpha \exp(-aw))+2(\alpha^2 a^2-1), \\
		& A'''(w)=-2\alpha a^3\exp(-aw)(1-4\alpha+\alpha a w+4\alpha\exp(-aw))\enspace .
	\end{align*}
	
	In this part of the proof we frequently use one of the implications of intermediate value theorem: if $f(x)$ and $f'(x)$ are two continuous and bounded functions, then there exists a root of $f'(x)$ between every two roots of $f(x)$. This also implies that the number of the roots of $f(x)$ is at most one plus the number of the roots of $f'(x)$.

	\begin{figure}[h!]
		\begin{center}
			\includegraphics[scale=0.15]{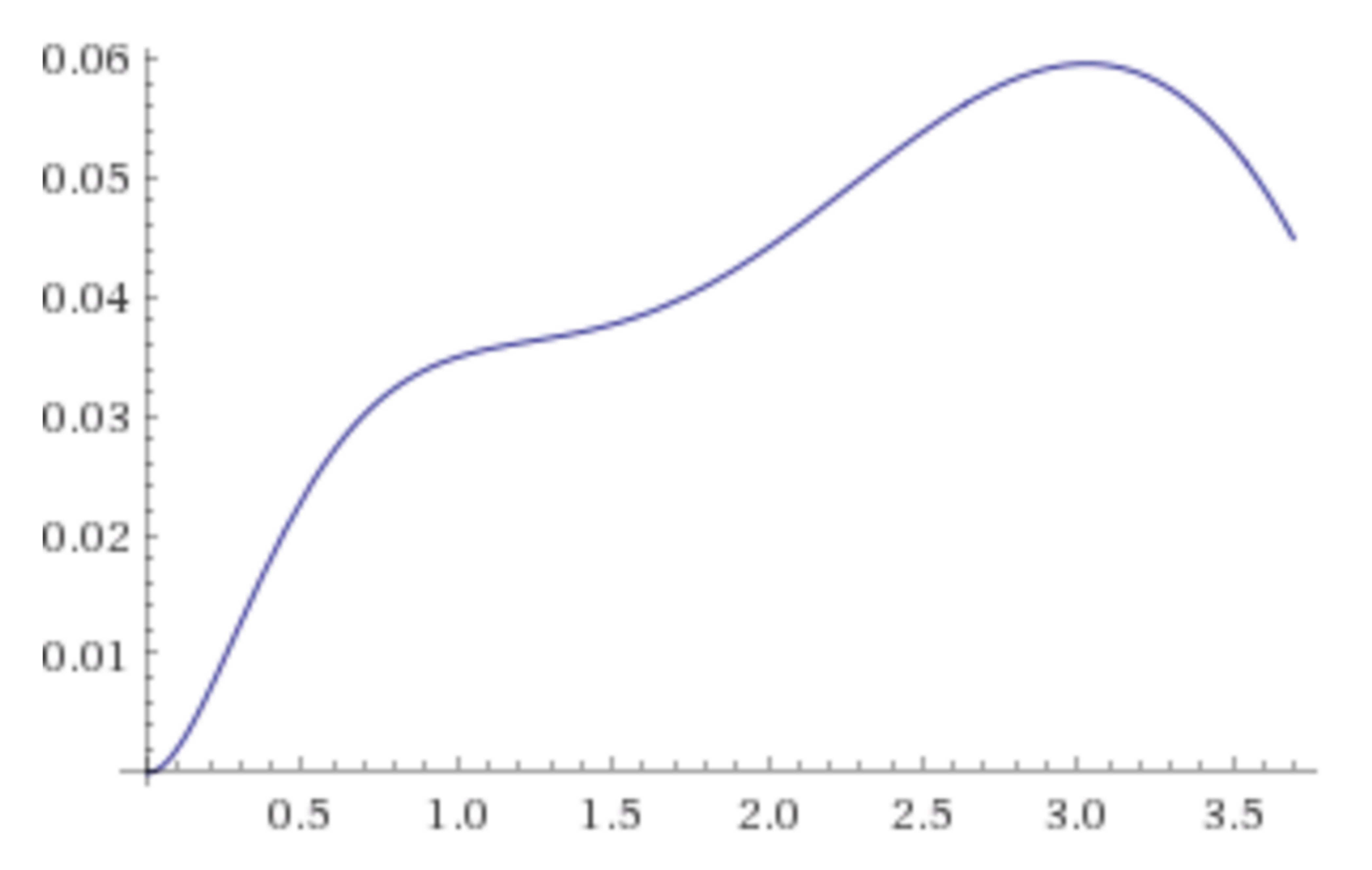}
		\end{center}
		\caption{\small{The plot shows function $A(w)$ for values of $w$ from $0$ to $\tan(a)\approx 3.7$.}}
		\label{fig:fig1}
	\end{figure}
	
	We claim that $A'''(w)$ has at most two roots. The reason for this is because $-2\alpha a^3\exp(-aw)$ is always non-zero, and $1-4\alpha+\alpha a w+4\alpha\exp(-aw)$ has at most two roots, because its derivative, $\alpha a(1-4\exp(-aw))$ has exactly one root, which is $\ln(4)/a$. 
		
	The fact that $A'''(w)$ has at most two roots implies that $A''(w)$ has at most three roots, which are $w_1\approx 0.28157$, $w_2\approx 1.24251$, and $w_3\approx 2.27082$. We note that $A'(w)$ is positive at all these points. Therefore $A'(w)$ has at most two roots, because otherwise there would be a point in which $A'(w)\leq 0$ and $A''(w)=0$ which is impossible.
		
	Note that $A'(0)=0$, therefore $A'(w)$ has at most one root in $\mathcal{R}^+$. Now we note that $A(0^+)>0$ because $A'(0)=0$ and $A''(0)=2(\alpha a^2(1-\alpha)+\alpha^2 a^2-1)>0$. Also $A(\tan(a))>0$. Now if $A(w)<0$ for some $0<w<\tan(a)$, then $A(w)$ would have at least two roots in range $(0,\tan(a))$ which results in $A'(w)$ having two roots in $\mathcal{R}^+$. Since this is not true, we have $A(w)\geq 0$ for every $0\leq w\leq \tan(a)$.
\end{proof}	

\subsection{Proof of Lemma \ref{lemma:partition}}

\begin{proof} 
	Let $X_1,\ldots,X_n$ be random variables representing the items, and let $Y_1,\ldots,Y_m$ be iid random variables with distribution function $F(x)=\prod_{i=1}^{k}F_i(x)$. For the expectation of the maximum of these variables we have:
	\begin{align}
	\E[\max_{i=1}^{k} Y_i]
	&= \int_{0}^{\infty} Pr[\max_{i=1}^{k} Y_i\geq x] dx \nonumber \\
	&= \int_{0}^{\infty} \bigg(1-\prod_{i=1}^{m}F(x)\bigg) dx \nonumber \\
	&= \int_{0}^{\infty} \bigg(1-\prod_{i=1}^{m}\prod_{j=1}^{k}F_j(x)\bigg) dx \nonumber \\
	&= \int_{0}^{\infty} \bigg(1-\prod_{i=1}^{n}F_i(x)\bigg) dx \nonumber \\
	&= \int_{0}^{\infty} Pr[\max_{i=1}^{n} X_i\geq x] dx \nonumber \\
	&= \E[\max_{i=1}^{n}X_i] \enspace .\label{ineq:optsAreEq}
	\end{align}
	This shows that the optimum solution is the same for both sets of items. Let $\tau_Y$ and $\tau_X$ be random variables that denotes the index of the picked items in $Y_1,\ldots,Y_m$ and $X_1\ldots,X_n$ respectively. Theorem \ref{theorem:iidResult} states that there exists some $s$ such for every $m\geq s$, we have $E[Y_{\tau_Y}] \geq (1-\epsilon/2)\alpha E[\max_{i=1}^{m}Y_i]$. In the following we show that there exist some $m_2$ such that for every $m\geq m_2$, $E[X_{\tau_X}]\geq (1-\epsilon/2)E[Y_{\tau_Y}]$. This proves the lemma for every $m\geq m_\epsilon=\max\{s,m_2\}$. In other words,
	\begin{align*}
		\E[X_{\tau_X}] &\geq (1-\epsilon/2)\E[Y_{\tau_Y}] \\ &\geq (1-\epsilon/2)^2\alpha \E[\max_{i=1}^{m}Y_i] \\ &\geq (1-\epsilon)\alpha \E[\max_{i=1}^{m}Y_i] \\ &=(1-\epsilon)\alpha \E[\max_{i=1}^{n}X_i]\enspace .
	\end{align*}
	
	Since $\E[Z]=\int_{0}^{\infty}Pr[Z\geq z]dz$ for every non-negative random variable $Z$, we show $Pr[X_{\tau_X}\geq x]\geq(1-\epsilon/2) Pr[Y_{\tau_Y}\geq x]$ for every $x\geq 0$ in order to prove $\E[X_{\tau_X}]\geq (1-\epsilon/2)\E[Y_{\tau_Y}]$.
	
	In the following, we use $G_i(x)$ to denote $1-F_i(x)$. For every non-negative $x$ we have:
	\begin{align}
	Pr[X_{\tau_X}\geq x]
	&= \sum_{i=1}^{n}Pr[X_{\tau_X}\geq x|\tau_X=i]Pr[\tau_X=i] \nonumber \\
	&= \sum_{i=1}^{n}\prod_{j=1}^{i-1}F_j(\theta_{\lceil j/k \rceil}) G_i(\max\{x,\theta_{\lceil i/k \rceil}\})\enspace .\nonumber
	\end{align}
	By rewriting the above sum with respect to the $m$ partitions we have:
	\begin{align}
	Pr[X_{\tau_X}\geq x]
	&= \sum_{i=0}^{m-1}\sum_{j=1}^{k}\prod_{l=1}^{ik+j-1}F_l(\theta_{\lceil l/k\rceil})G_{ik+j}(\max\{x,\theta_{i+1}\}) \nonumber \\
	&= \sum_{i=0}^{m-1}\sum_{j=1}^{k}\prod_{l=1}^{ik}F_l(\theta_{\lceil l/k\rceil})\prod_{p=1}^{j-1}F_{ik+p}(\theta_{i+1})G_{ik+j}(\max\{x,\theta_{i+1}\}) \nonumber \\
	&= \sum_{i=0}^{m-1}\prod_{l=1}^{ik}F_l(\theta_{\lceil l/k\rceil})\sum_{j=1}^{k}\prod_{p=1}^{j-1}F_{ik+p}(\theta_{i+1})G_{ik+j}(\max\{x,\theta_{i+1}\}) \enspace . \label{ineq:messy}
	\end{align}
	
	Note that $X_1,\ldots,X_n$ are $m$-partitioned, hence for every partition $0\leq i< m$ and $x\geq 0$ we have $\prod_{l=ik+1}^{ik+k}F_l(x)=F(x)$. Therefore $\prod_{l=1}^{ik}F_l(\theta_{\lceil l/k\rceil})=\prod_{l=1}^{i}F(\theta_{l})$. By this replacement, Inequality \eqref{ineq:messy} can be written as follows:
	\begin{align}
	Pr[X_{\tau_X}\geq x]
	&= \sum_{i=0}^{m-1}\prod_{l=1}^{i}F(\theta_l)\sum_{j=1}^{k}\prod_{p=1}^{j-1}F_{ik+p}(\theta_{i+1})G_{ik+j}(\max\{x,\theta_{i+1}\}) \enspace . \label{ineq:messy2}
	\end{align}
	
	Moreover, for every $0\leq i<m$ and $1\leq j\leq k$ we have:
	\begin{align}
	\prod_{p=1}^{j-1}F_{ik+p}(\theta_{i+1})
	&\geq \prod_{p=1}^{k}F_{ik+p}(\theta_{i+1}) & \text{every $F_t(x)$ is a most 1}\nonumber \\
	&= F(\theta_{i+1}) & \nonumber \\
	&= 1-G(\theta_{i+1}) & \nonumber \\
	&\geq 1-\frac{a\tan(a)}{m} & \text{Inequality \ref{ineq:useful} for } h(s)=\cos(as) \nonumber \\
	&\geq 1-\frac{\epsilon}{2} & \text{for every } m\geq m_2= \frac{2a\tan(a)}{\epsilon}\enspace.\label{ineq:prodSimple}
	\end{align}
	
	Inequality \eqref{ineq:prodSimple} shows that for a large enough $m$, the left hand side of the inequality becomes close enough to 1. By using this inequality in Inequality \eqref{ineq:messy2} we have:

	\begin{align}
	Pr[X_{\tau_X}\geq x] &\geq \sum_{i=0}^{m-1}\prod_{l=1}^{i}F(\theta_l)\sum_{j=1}^{k}(1-\frac{\epsilon}{2})G_{ik+j}(\max\{x,\theta_{i+1}\}) \nonumber \\
	&= (1-\frac{\epsilon}{2})\sum_{i=0}^{m-1}\prod_{l=1}^{i}F(\theta_l)\sum_{j=1}^{k}G_{ik+j}(\max\{x,\theta_{i+1}\}) \label{ineq:messy3} \enspace .
	\end{align}
	
	Let $r=\max\{x, \theta_{j+1}\}$. By multiplying every term in Inequality \eqref{ineq:messy3} by $\prod_{p=1}^{j-1}F_{ik+p}(r)$, which is less than or equal to 1, we have:	
	\begin{align}
	Pr[X_{\tau_X}\geq x]
	&\geq (1-\frac{\epsilon}{2})\sum_{i=0}^{m-1}\prod_{l=1}^{i}F(\theta_l)\sum_{j=1}^{k}G_{ik+j}(r) & \nonumber \\ 
	&\geq (1-\frac{\epsilon}{2})\sum_{i=0}^{m-1}\prod_{l=1}^{i}F(\theta_l)\sum_{j=1}^{k}\prod_{p=1}^{j-1}F_{ik+p}(r)G_{ik+j}(r) & \nonumber \\
	&=	(1-\frac{\epsilon}{2})\sum_{i=0}^{m-1}\prod_{l=1}^{i}F(\theta_l)\sum_{j=1}^{k}\prod_{p=1}^{j-1}F_{ik+p}(r)(1-F_{ik+j}(r)) & \nonumber \\
	&=	(1-\frac{\epsilon}{2})\sum_{i=0}^{m-1}\prod_{l=1}^{i}F(\theta_l)\sum_{j=1}^{k}\bigg(\prod_{p=1}^{j-1}F_{ik+p}(r)-\prod_{p=1}^{j}F_{ik+p}(r)\bigg) \nonumber
	\end{align}
	
	Note that the inner sum forms a telescoping series, hence we can simplify it as follows:
	\begin{align}
	Pr[X_{\tau_X}\geq x]
	&\geq	(1-\frac{\epsilon}{2})\sum_{i=0}^{m-1}\prod_{l=1}^{i}F(\theta_l)\bigg(1-\prod_{p=1}^{k}F_{ik+p}(r)\bigg) & \nonumber \\
	&=	(1-\frac{\epsilon}{2})\sum_{i=0}^{m-1}\prod_{l=1}^{i}F(\theta_l)\bigg(1-F(r)\bigg) & \nonumber \\
	&= (1-\frac{\epsilon}{2})\sum_{i=0}^{m-1}\prod_{l=1}^{i}F(\theta_l)G(r) \enspace . & \label{ineq:messy4}
	\end{align}
	
	Note that $Pr[Y_{\tau_Y}\geq x]=\sum_{i=0}^{m-1}\prod_{l=1}^{i}F(\theta_l)G(\max\{x,\theta_{i+1}\})$. Using this in Inequality \eqref{ineq:messy4} results that $Pr[X_{\tau_X}\geq x]\geq (1-\epsilon/2)Pr[Y_{\tau_Y}\geq x]$, hence the proof is complete.
\end{proof}

\subsection{Proof of Lemma \ref{lemma:discard}}
\begin{proof} 
	Let $\rho$ and $\rho'$ be random variables that denote the index of the maximum with the smallest index amongst $X_1,\ldots,X_n$ and $X_{S_1},\ldots,X_{S_r}$, respectively. Then we have
	\begin{align*}
		\E[\max_{i\in S}X_i]
		&= \sum_{i\in S}\E[X_i|i=\rho']Pr[i=\rho'] \\
		&\geq \sum_{i\in S}\E[X_i|i=\rho]Pr[i=\rho] \\
		&\geq \frac{p}{k}\sum_{i=1}^{n}\E[X_i|i=\rho]Pr[i=\rho] \\
		&= \frac{p}{k}\E[\max_{1\leq i\leq n}X_i]
	\end{align*}
	which completes the proof.
\end{proof}

\subsection{Proof of Theorem \ref{theorem:hardness}}

\begin{proof} 
	Pick arbitrary numbers $m$ and $\epsilon$. Suppose we have $2m$ distribution. Each of the first $m$ distributions gives $1$ with probability $1$. Each of the last $m$ distributions gives $0$ with probability $(1-\epsilon)^{1/m}$ and $\frac{1}{\epsilon}$ otherwise. Notice that with probability $\Big((1-\epsilon)^{1/m}\Big)^m = 1-\epsilon$ all of the last $m$ items are $0$ and with probability $\epsilon$ at least one of the last $m$ items is $\frac 1 {\epsilon}$. Hence in expectation the optimum takes $$1\times(1-\epsilon)+\epsilon\times\frac 1 {\epsilon}=2-\epsilon.$$
	
	While any online algorithm takes at most $\max(1,\epsilon\times\frac 1 {\epsilon})=1.$ Therefore, the approximation factor of any online algorithm is upper-bounded by $$\frac{1}{2-\epsilon} =\frac{2+\epsilon}{4-\epsilon^2}\leq \frac {2+\epsilon}{4}\leq  0.5 + \epsilon.$$
\end{proof}

\end{document}